\documentclass[12pt]{article}
\usepackage[utf8]{inputenc}
\usepackage{amsmath}
\usepackage{amssymb}
\usepackage{amsthm}
\usepackage{graphicx}
\usepackage{caption}
\usepackage{subcaption}
\usepackage{amsfonts}
\usepackage{verbatim}
\usepackage{abstract}

\usepackage{indentfirst}

\renewcommand{\leq}{\leqslant}
\renewcommand{\geq}{\geqslant}

\newtheorem{theorem}{Theorem}
\newtheorem{lemma}{Lemma}

\newtheorem{corollary}{Corollary}

\title{On computational complexity of length embeddability of graphs}
\author{Mikhail Tikhomirov}
\date{}

\begin{document}

\maketitle

\begin{abstract}
A graph $G$ is embeddable in $\mathbb{R}^d$ if vertices of $G$ can be assigned with points of $\mathbb{R}^d$ in such a way that all pairs of adjacent vertices are at the distance 1. We show that verifying embeddability of a given graph in $\mathbb{R}^d$ is NP-hard in the case $d > 2$ for all reasonable notions of embeddability.
\end{abstract}

\section{Introduction}

The \emph{distance graph of $S \subset \mathbb{R}^d$} is defined as the graph $G = (V, E)$, where $V = S$ and $E$ is the set of all pairs of points $x, y \in S$ such that $x$ and $y$ are at the distance 1. A graph is a \emph{distance graph in $\mathbb{R}^d$} if it is isomorphic to the distance graph of some set $S \subset \mathbb{R}^d$. Some famous problems concerning distance graphs are the Erd\H{o}s' unit distance problem on the maximal number of unit distances between $n$ points in $\mathbb{R}^2$ (see \cite{brass2005research}, \cite{erdos1946sets}, \cite{chilakamarri1993unit}),
the Hadwiger--Nelson problem on the chromatic number of $\mathbb{R}^2$ (see \cite{brass2005research}, \cite{de1951colour}, \cite{raigorodskii2001borsuk}), etc.; surveys of various results about distance graphs can be found at \cite{raigorodskii2013coloring}, \cite{raigorodskii2014cliques}.

We also consider a similar notion of \emph{embeddability in $\mathbb{R}^d$} (see, e.g., \cite{saxe1980embeddability}).
A graph $G = (V, E)$ is \emph{embed\-dable in $\mathbb{R}^d$} if there exists a mapping $\varphi: V \rightarrow \mathbb{R}^d$ such that
$||\varphi(u) - \varphi(v)||_{\mathbb{R}^d} = 1$ for all pairs $(u, v) \in E$. It is clear that any distance graph in $\mathbb{R}^d$
is embeddable in $\mathbb{R}^d$ but the converse does not always hold. These two notions differ in the following:

\begin{itemize}

\item Different vertices of an embeddable graph may be assigned with the same point in $\mathbb{R}^d$ while all vertices of a distance graph should be assigned with pairwise distinct points.

\item Non-adjacent vertices of an embeddable graph can be located at the distance 1 while non-adjacent vertices of a distance graph are forbidden to be placed at distance~1.

\end{itemize}

We will say that an embedding $\varphi: V \rightarrow \mathbb{R}^d$ is \emph{strict} if $\forall u, v \in V$                                     
$(u, v) \in E \iff ||\varphi(u) - \varphi(v)||_{\mathbb{R}^d} = 1$; we will say that an embedding $\varphi: V \rightarrow \mathbb{R}^d$ is \emph{injective} if $\forall u, v \in V$
$v \neq u \Rightarrow \varphi(v) \neq \varphi(u)$. It is clear that a graph $G$ is a distance graph in $\mathbb{R}^d$ iff there exists a strict and injective embedding of $G$ in $\mathbb{R}^d$. Thus we obtain four different notions of embeddability (strict/non-strict, injective/non-injective) which include two notions described above.

For each of the four notions of embeddability in $\mathbb{R}^d$ we can pose the computational decision problem of determining
embeddability of chosen type for the given graph; we shall call this problem
\textbf{$\mathbb{R}^d$-UNIT-DISTANCE-(STRICT)-(INJECTIVE)-EM\-BED\-DABILITY} 
depending on the em\-beddability type. The computational complexity of these problems is studied in \cite{saxe1980embeddability}, \cite{horvat2011computational}.
In \cite{horvat2011computational} it is shown that
\textbf{$\mathbb{R}^d$-UNIT-DISTANCE-(STRICT)-(INJECTIVE)-EM\-BED\-DA\-BILITY} is NP-hard for each type of em\-beddability and each value of $d \geq 2$.
Unfortunately, the proof in \cite{horvat2011computational} for the case $d > 2$ is false as it is based on the result \cite{lovasz1983self} due to Lov{\'a}sz
which states the upper bound $d+1$ for the chromatic number of the $d$-dimensional sphere circumscribed about a regular simplex on $d+1$ vertices with unit length edges.
In \cite{raigorodskii2010chromatic}, \cite{raigorodskii2012chromatic} Raigorodskii points out that this bound is wrong and proves
an exponential lower bound of this value; thus a new proof is needed for the case $d > 2$, which is the point of this paper.

The main result is

\begin{theorem}

Computational problems $\mathbb{R}^d$-UNIT-DISTANCE-EMBEDDABILITY, $\mathbb{R}^d$-UNIT-DISTANCE-STRICT-EMBEDDABILITY,
$\mathbb{R}^d$-UNIT-DISTANCE-INJECTIVE-EMBEDDABILITY, $\mathbb{R}^d$-UNIT-DISTANCE-STRICT-INJECTIVE-EMBEDDABI\-LI\-TY
are NP-hard for each $d > 2$.

\end{theorem}

To prove this result we contrust a reduction of the classic NP-complete problem of graph vertex 3-coloring (\textbf{3-COLORING})
(see \cite{karp1972reducibility}) to each of the four embeddability problems: for any given graph $G$ we explicitly construct a graph $H = $
3-COLORING-$\mathbb{R}^d$-UNIT-DISTANCE-EMBEDDABILITY-RE\-DUC\-TION$(G)$ such that the size of $H$ is linear in the size of $G$ (for every fixed $d$) and the following conditions hold:

\begin{itemize}

\item If no valid vertex 3-coloring of $G$ exists, then there is no embedding of $H$ in $\mathbb{R}^d$;

\item If a valid vertex 3-coloring of $G$ exists, then there is a strict injective embedding of $H$ in $\mathbb{R}^d$.

\end{itemize}

The possibility of such construction implies NP-hardness
of all four mentioned problems. It should be mentioned that the question whether the described problems lie in NP is open.

\section{Notion of rod}

Let us introduce some necessary definitions.

A \emph{weighted graph} $G = (V, E, w)$ is an ordered triple such that $(V, E)$ is a graph and $w: E \rightarrow \mathbb{R}_+$ is a function that assigns a positive number to each element of $E$; for every edge $e \in E$ we will say that $w(e)$ is the \emph{length} of the edge $e$.
If $w \equiv 1$, the weighted graph $G$ is called a \emph{unit distance graph}.
A \emph{length embedding} (or, more simply, an \emph{embedding}) of the weighted graph $G = (V, E, w)$ in $\mathbb{R}^d$ is a map
$\varphi: V \rightarrow \mathbb{R}^d$ such that $\forall u, v \in V$
$(u, v) \in E \Rightarrow ||\varphi(u) - \varphi(v)||_{\mathbb{R}^d} = w((u, v))$.

\emph{\textbf{Remark}: In the sequel, we will identify vertices of the graph with points of $\mathbb{R}^d$~--- their images under the embedding if that doesn't cause confusion.}

An embedding $\varphi$ of the weighted graph $G = (V, E, w)$ in $\mathbb{R}^d$ is called \emph{non-critical} if the following conditions hold:

\begin{itemize}

\item $\forall u, v \in V$ $u \neq v \Rightarrow \varphi(u) \neq \varphi(v)$ (no two vertices are at the same point);

\item $\forall u, v \in V$ $(u, v) \notin E \Rightarrow ||\varphi(u) - \varphi(v)||_{\mathbb{R}^d} \neq 1$ (no two non-adjacent vertices are at the distance 1);

\item no three vertices lie on a (one-dimensional) straight line.

\end{itemize}

Consider a weighted graph $G = (V, E, w)$ and a pair of its vertices $u, v \in V$. The graph $G$ is called a \emph{($d$-dimensional) $(u, v)$-rod of length $l$} if the following conditions hold:

\begin{itemize}

\item the distance between vertices $u$ and $v$ is equal to $l$ in each embedding of $G$ in $\mathbb{R}^d$;

\item there exists a non-critical embedding of $G$ in $\mathbb{R}^d$.

\end{itemize}

If a unit distance graph $G$ is also a $d$-dimensional $(u, v)$-rod of length $l$, we call $G$ a \emph{$d$-dimensional unit distance $(u, v)$-rod of length $l$}.

A weighted graph $G = (V, E, w)$ is called a \emph{(unit distance) $d$-dimensional rod of length $l$} if there exist two vertices $u, v \in V$ such that $G$ is a (unit distance) $d$-dimensional $(u, v)$-rod of length $l$.

We suppose that $d$ is a fixed constant throughout the whole paper, thus in the sequel we will write ``rod'' for ``$d$-dimensional rod''.

\begin{lemma}
Let $G = (V_G, E_G, w_G)$, $H = (V_H, E_H, w_H)$ be weighted graphs. Suppose $V_G \cap V_H = \{u, v\}$, $e = (u, v) \in E(G)$, $w_G(e) = l$,
and $H$ is a $(u, v)$-rod of length $l$. Let $G'~=~(V_{G'}, E_{G'}, w_{G'})$, where $V_{G'} = V_G \cup V_H$,
$E_{G'} = (E_G \backslash \{e\}) \cup E_H$, 
$w_{G'} = w_G I(E_G \backslash \{e\}) + w_H I(E_H)$ (informally, we replace the edge $e$ in $G$ by the subgraph $H$ to obtain $G'$). Then:

\begin{itemize}

\item If there is no embedding of $G$ in $\mathbb{R}^d$, then there is no embedding of $G'$ in $\mathbb{R}^d$.

\item If there exists a \textbf{non-critical} embedding of $G$ in $\mathbb{R}^d$, then there exists a 
\textbf{non-critical} embedding of $G'$ in $\mathbb{R}^d$.

\end{itemize}
\end{lemma}

\begin{proof}

In any embedding of $G'$ the distance between vertices $u$ and $v$ is equal to $l$. Suppose we have an embed\-ding of $G'$; we can erase all vertices outside $V_G$ to obtain an embedding of $G$. The first claim is thus proven.

Now consider a non-critical embedding $\varphi_G$ of the weighted graph $G$ in $\mathbb{R}^d$. Construct an embedding $\varphi_{G'}$ of $G'$ as follows:

\begin{itemize}
\item Let $\varphi_{G'}(x) = \varphi_G(x)$ for all $x \in V_G$;

\item Choose a non-critical embedding $\varphi_H$ of the weighted graph $H$ such that $\varphi_H(u) = \varphi_G(u)$, $\varphi_H(v) = \varphi_G(v)$ (such embedding exists since $||\varphi_G(u) - \varphi_G(v)|| = l$ and $H$ is a $(u, v)$-rod of length $l$); let $\varphi_{G'}(y) = \varphi_H(y)$ for all $y \in V_H$.

\end{itemize} 

It is clear that this definition of $\varphi_{G'}$ is consistent. However, it is possible that $\varphi_{G'}$ is not a non-critical embedding. Note
that no vertex of $V_{G'} \backslash \{u, v\}$ lies on the straight line $uv$ since the embeddings $\varphi_G$ and $\varphi_H$ are non-critical.

Let $S$ denote the set of all rotations of $\mathbb{R}^d$ about the line $uv$. $S$ is isomorphic to the $(d-2)$-dimensional sphere 
(each rotation can be assigned with the image of some point which doesn't lie on $uv$). For any $\psi \in S$ let $\psi *_H \varphi_{G'}$ denote the mapping from $V_{G'}$ in $\mathbb{R}^d$ such that $\psi *_H \varphi_{G'}(x) = \varphi_{G'}(x)$ for every $x \in V_G$ and $\psi *_H \varphi_{G'}(y) = \psi(\varphi_{G'}(y))$ for every $y \in V_H$; clearly, this definition is consistent. It is also clear that for every rotation $\psi \in S$ the mapping $\psi *_H \varphi_{G'}$ is an embedding of $G'$ in $\mathbb{R}^d$.

We now show that there exists a rotation $\psi \in S$ such that $\psi *_H \varphi_{G'}$ is a non-critical embedding of $G'$ in $\mathbb{R}^d$. Consider all $\psi \in S$ such that the embedding $\psi *_H \varphi_{G'}$ is not non-critical for some reason. In that case, one of the following conditions must hold:

\begin{itemize}

\item The embedding $\psi *_H \varphi_{G'}$ places two vertices of $G'$ (denote them $x$ and $y$) at the same point. It follows from the non-criticality of $\varphi_G$ and $\varphi_H$ that $x$ and $y$ cannot lie both in $V_G$ or both in $V_H$.
Thus WLOG $x\in V_G \backslash \{u, v\}$, $y \in V_H \backslash \{u, v\}$.

The vertex $y$ does not lie on the line $uv$ and no two rotations place $y$ at the same point. Therefore for every pair of vertices $x, y$ there is at most one rotation $\psi \in S$ that superposes $x$ and $\psi(y)$, thus the set of all rotations $\psi$ such that the embedding $\psi *_H \varphi_{G'}$ places some two vertices in the same point is finite and its spherical measure in $S$ is zero.

\item The embedding $\psi *_H \varphi_{G'}$ places two non-adjacent vertices of $G'$ (denote them $x$ and $y$ once more) at the distance 1. Once again, $x, y \in V_G$ or $x, y \in V_H$ leads to a contradiction; thus WLOG $x\in V_G \backslash \{u, v\}$, $y \in V_H \backslash \{u, v\}$.

Let $P_y$ denote the $(d - 2)$-dimensional sphere~--- the locus of the point $\psi(y)$ for all $\psi \in S$; the radius of $P_y$ is non-zero since $y$ does not lie on the line $uv$. If $||x - \psi(y)||_{\mathbb{R}^d} = 1$, then $\psi(y)$ lies on the $(d - 1)$-dimensional sphere of radius 1 centered at $x$; denote it $P_x$. We assume that the intersection of $P_x$ and $P_y$ is not empty.

If $P_x$ contains $P_y$ as a subset, then $x$ must lie on the line $uv$; that would contradict the non-criticality of $\varphi_G$. Otherwise, the intersection of $P_x$ and $P_y$ is a $(d - 3)$-dimensional sphere (possibly, of zero radius).

In any case, the set of rotations that place $x$ and $\psi(y)$ at the distance 1 has zero measure in $S$. Thus the set of rotations that place some two non-adjacent vertices at the distance 1 has zero measure in $S$.

\item The embedding $\psi *_H \varphi_{G'}$ places some three vertices on a straight line; denote these vertices $x, y, z$. 
Similarly to previous cases, if we assume $x, y, z \in V_G$ or $V_H$ we arrive at a contradiction.

WLOG, let $x, y \in V_G \backslash \{u, v\}$, $z \in V_H \backslash \{u, v\}$. Since the point $z$ can not lie on the line $uv$, the sphere  $\psi(z)$ for $\psi \in S$ has non-zero radius and the line $xy$ passes through the point $\psi(z)$ for at most two values of $\psi$.

Now let $x \in V_G \backslash \{u, v\}$ and $y, z \in V_H \backslash \{u, v\}$. 
The rotation $\psi \in S$ places the point $x$ on the line $\psi(yz)$ iff the point $\psi^{-1}(x)$
lies on the line $yz$ (here $\psi^{-1}$ means the inverse rotation of $\psi$), therefore in this case the line $yz$ must cross the locus of $\psi^{-1}(x)$ for all $\psi \in S$. Clearly, the locus is a sphere of non-zero radius, thus line $\psi(yz)$
passes through the point $x$ for at most two values of $\psi$.

It follows from the above that the set of rotations $\psi \in S$ such that $\psi *_H \varphi_{G'}$ places some three vertices on a straight line is finite.

\end{itemize}

To sum up, the set of rotations $\psi$ such that the embedding $\psi *_H \varphi_{G'}$ is not non-critical has zero measure in the $(d - 2)$-dimensional sphere of all possible rotations about the line $uv$. Therefore almost every rotation $\psi \in S$ yields a non-critical embedding $\psi *_H \varphi_{G'}$ of the graph $G'$ in $\mathbb{R}^d$. 

\end{proof}

\section{Construction of rods}

Let $h = \sqrt{\frac{d + 1}{2d}}$ denote the altitude length of a regular $d$-dimensional simplex with the edge length 1; denote $D = 2h$. Clearly, $D > \sqrt{2}$.

\begin{lemma}

Let $G$ and $H$ be unit distance rods of length $a$ and $b$ respectively. Then there exists a unit distance rod of length $ab$.

\end{lemma}

\begin{proof}

It suffices to make lengths of all edges of $G$ be equal to $b$ and successively apply Lemma 1 to every edge of the resulting graph and the graph $H$.

\end{proof}

Consider a graph $M_d$ on a set of vertices $V_{M_d} = K_1 \cup K_2 \cup \{A, B, C\}$, $|K_1| = |K_2| = d$. Add the following edges of unit length to $M_d$:

\begin{itemize}

\item make cliques on $K_1$ and $K_2$;

\item connect the vertices $A$ and $B$ with every vertex of $K_1$;

\item connect the vertices $A$ and $C$ with every vertex of $K_2$;

\item finally, connect the vertices $B$ and $C$.

\end{itemize}

\begin{figure}[h]
\begin{centering}
\includegraphics{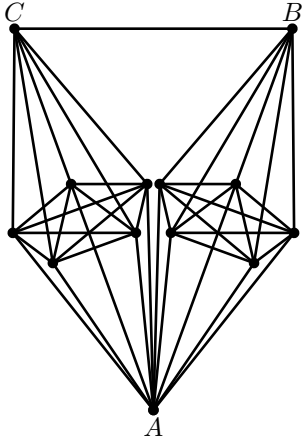}
\caption{5-dimensional Moser spindle}
\end{centering}
\end{figure}

The graph $M_d$ is called a \emph{$d$-dimensional Moser spindle} (the figure 1 illustrates a 5-dimensional Moser spindle). Is it easy to see that $M_d$ is a unit distance $d$-dimensional $(A, B)$-rod of length $D$.

Repeatedly applying Lemma 2 to copies of $M_d$, we arrive at

\begin{corollary}

For every non-negative integer $k$ there exists a unit distance $d$-dimensional rod of length $D^k$.

\end{corollary}

\begin{lemma}

For all numbers $a, b$ such that $0 < a < b < 1$ there exists a number $l$ satisfying $a < l < b$ and a graph $G$ such that $G$ is a unit-distance $d$-dimensional rod of length $l$.

\end{lemma}

\begin{proof}

Construct $G$ as follows. Choose a set of vertices $K$ of size $d-1$ and connect its elements pairwise by unit length edges. 
Then, take a sequence of vertices $v_1$,~\ldots,~$v_n$ (the exact number of vertices $n$ will be determined later) and connect every vertex of the sequence $v_i$ with every vertex of $K$ by a unit length edge. If the location of vertices of $K$ is fixed, then all vertices $v_1$, \ldots, $v_n$ must lie on some circle centered at $O$, where $O$ is the center of the regular simplex with vertices in $K$. The radius of the circle is equal to the altitude length of the $(d-1)$-face of the regular $d$-simplex with the side length 1, i.e. $\sqrt{\frac{d}{2(d-1)}} = r$. Let $\pi$ denote the plane containing this circle.

For each $i$ from 1 to $n - 1$ connect the vertices $v_i$ and $v_{i+1}$ by an edge of length 1; also for each $i$ from 1 to $n-2$ connect the vertices $v_i$ and $v_{i + 2}$ by an edge of length $D$. Now in every embedding of the graph $G$ the angle $\angle v_i O v_{i+1}$ is equal to the dihedral angle of a regular $d$-simplex; denote this angle $\alpha = \arccos \frac{1}{d}$. Additionally, the least rotation of the plane $\pi$ about the point $O$ that moves the point $v_i$ to $v_{i+1}$ has the same direction for every $i$. It is clear that no three vertices of $G$ lie on a straight line.

Let us introduce an angular coordinate system $\psi$ on $\pi$ centered at $O$ such that $\psi(v_1) = 0$, $\psi(v_2) = \alpha$. 
Clearly, $\psi(v_i) = (i - 1) \alpha \mod 2 \pi$
(by $\alpha \mod 2 \pi$ we mean $\alpha + k \times 2 \pi$ for an integer $k$ such that
$0 \leq \alpha + k \times 2\pi < 2\pi$). By Niven's theorem (see \cite{niven1956irrational}, Corollary 3.12), $\alpha / 2 \pi$ can not 
be a rational number when $d \geq 3$, therefore the infinite sequence $x_i = (i - 1) \alpha \mod 2 \pi$ is dense in $[0; 2 \pi]$.
Thus there exists a positive integer $N$ such that $x_N \in (2 \arcsin {\frac{a}{2r}}; 2 \arcsin{\frac{b}{2r}})$ and $||v_1 - v_N|| \in (a; b)$.

It follows from the above that the graph $G$ is a $d$-dimensional $(v_1, v_N)$-rod. Finally, successively apply Lemma 2 to each $D$-length edge of the graph $G$ and the graph $M_d$; the resulting graph is a unit distance $d$-dimensional $(v_1, v_N)$-rod that satisfies all the conditions.

\end{proof}


\begin{theorem} For all numbers $a, b$ such that $0 < a < b$ there exists a number $l$ satisfying $a < l < b$ and a graph $G$ such that $G$ is a unit-distance $d$-dimensional rod of length $l$.

\end{theorem}

\begin{proof}

Choose a non-negative integer $k$ such that $D^k > b$ and denote $G'$ the rod obtained by applying Lemma 3 for numbers $\frac{a}{D^k}$ and $\frac{b}{D^k}$. Now apply Lemma 2 to the graph $G'$ and the rod of length $D^k$.

\end{proof}

Let $\text{RodLength}(a, b)$ denote the number $l$ produced by Theorem 2 for given numbers $a$ and $b$,
and $\text{Rod}(a, b)$ denote the rod of corresponding length.

\section{The reduction setup}

Consider a graph $G = (V_G, E_G)$~--- the input of the \textbf{3-COLORING} problem. We now construct a weighted graph $H = (V_H, E_H, w_H)$ = 3-COLORING-$\mathbb{R}^d$-EMBEDDABI\-LI\-TY-RE\-DUCTION($G$) such that the embeddability of $H$ in $\mathbb{R}^d$ is equivalent to the existence of a solution to the \textbf{3-COLORING} for the graph $G$. We shall idenitify the elements of $V_G$ and the integers from 1 to $|V_G|$ for the sake of convenience.


To establish properties of the following setup we will need the following

\begin{lemma}

Let $0 \leq l < L \leq R < r$, $\delta = \min(L - l, r - R)$.\\
Let also $G = (V_G, E_G, w_G)$, $H = (V_H, E_H, w_H)$~--- weighted graphs,\\
$v, u \in V_G$, $V_H = V_G \sqcup \{z\}$, $E_H = E_G \sqcup \{(v, z), (u, z)\}$,\\
$w_H(e) = w_G(e)$ for all $e \in E_G$,\\
$w_H((v, z)) = a \in \left(\frac{L + R}{2} - \frac{\delta}{3}; \frac{L + R}{2} + \frac{\delta}{3}\right)$,\\
$w_H((u, z)) = b \in \left(\frac{R - L}{2} + \frac{\delta}{3}; \frac{R - L}{2} + \frac{\delta}{2}\right)$.\\
Then:

\begin{itemize}

\item In every embedding of the graph $H$ the inequalities $l < ||v - u|| < r$ hold.

\item If there exists a \emph{non-critical} embedding of $G$ such that $L \leq ||v - u|| \leq R$, then there exists a \emph{non-critical} embedding of $H$.

\end{itemize}

\end{lemma}

\begin{proof}

First of all, let us show that $l < a - b < L \leq R < a + b < r$. Indeed,
\[a - b >\left (\frac{L + R}{2} - \frac{\delta}{3}\right) - \left(\frac{R - L}{2} + \frac{\delta}{2}\right) > L - \delta \geq l;\]
\[a - b < \left(\frac{L + R}{2} + \frac{\delta}{3}\right) - \left(\frac{R - L}{2} + \frac{\delta}{3}\right) = L;\]
\[a + b > \left(\frac{L + R}{2} - \frac{\delta}{3}\right) - \left(\frac{R - L}{2} + \frac{\delta}{3}\right) = R;\]
\[a + b < \left(\frac{L + R}{2} +\frac{\delta}{3}\right) + \left(\frac{R - L}{2} + \frac{\delta}{2}\right) < R + \delta \leq r.\]

Consider any embedding of the graph $H$. It follows from the triangle inequality applied to vertices $v, z, u$ that $l < a - b \leq ||u - v|| \leq a + b < r$. The first claim is thus proven.

Now, consider a non-critical embedding of the graph $G$ such that $||u - v|| \in [L; R]$. It follows from $a - b < L \leq ||u - v|| \leq R < a + b$ that it is possible to place the vertex $z$ in such a way that $||v - z|| = a$, $||u - z|| = b$ and $z$ does not lie on the line $vu$. We have obtained an embedding of the graph $H$; it is possible to modify this embedding to obtain a non-critical embedding by choosing an appropriate rotation of $z$ about the line $vu$; the proof of the existence of such rotation copies the similar proof from Lemma 1 almost entirely.

\end{proof}

Denote $r_0 = \sqrt{\frac{d}{2(d - 1)}}$, $\text{chord}(\alpha) = 2r_0 \sin{\alpha / 2}$~--- the length of the chord which contracts an $\alpha$-measured arc of a circle of radius $r_0$, $\varepsilon = \frac{\pi}{24}$.

Let us introduce additional notation as follows:

$\delta_{uu} = \min(\text{chord}(2\pi/3) - \text{chord}(2\pi / 3 - \varepsilon / 2), \text{chord}(2\pi / 3 + \varepsilon / 2) - \text{chord}(2\pi/3))$,

$a_{uu} = \text{RodLength}(\text{chord}(2\pi/3) - \delta_{uu}/3, \text{chord}(2\pi/3) + \delta_{uu}/3)$,

$b_{uu} = \text{RodLength}(\delta_{uu}/3, \delta_{uu}/2)$,

$\delta_{uv} = \text{chord}(\pi/3 - \varepsilon / 2) - \text{chord}(\pi/3 - \varepsilon)$,

$a_{uv} = \text{RodLength}\left(\frac{2r_0 + \text{chord}(\pi/3 - \varepsilon / 2)}{2} - \delta_{uv}/3, \frac{2r_0 + \text{chord}(\pi/3 - \varepsilon / 2)}{2} + \delta_{uv}/3\right)$,

$b_{uv} = \text{RodLength}\left(\frac{2r_0 - \text{chord}(\pi/3 - \varepsilon / 2)}{2} + \delta_{uv}/3, \frac{2r_0 - \text{chord}(\pi/3 - \varepsilon / 2)}{2} + \delta_{uv}/2\right)$,

$\delta_{vv} = \text{chord}\left(\frac{2}{3}\pi - \varepsilon) - \text{chord}(\frac{5}{2}\varepsilon\right)$,

$a_{vv} = \text{RodLength}\left(\frac{2r_0 + \text{chord}(2\pi / 3 - \varepsilon)}{2} - \delta_{vv}/3, \frac{2r_0 + \text{chord}(2\pi / 3 - \varepsilon)}{2} + \delta_{vv}/3\right)$,

$b_{vv} = \text{RodLength}\left(\frac{2r_0 - \text{chord}(2\pi / 3 - \varepsilon)}{2} + \delta_{vv}/3, \frac{2r_0 - \text{chord}(2\pi / 3 - \varepsilon)}{2} + \delta_{vv}/2\right)$.

\bigskip

Construct $H$ = 3-COLORING-$\mathbb{R}^d$-EMBEDDABILI\-TY-REDUCTION($G$) as follows: \[V_H = K \cup U \cup V \cup Aux, E_H = E_K \cup E_{KU} \cup E_{KV} \cup E_U \cup E_{VU} \cup E_V.\] Here:

\begin{itemize}

\item $Aux$ is the set of all auxiliary vertices used in the sequel of the description ($aux_{\ldots}$);

\item $K$~--- the set of vertices of size $d-1$;
\\ $E_K$~--- the set of edges connecting all pairs of vertices of $K$;

\item $U = \{u_0, u_1, u_2\}$;\\
$E_{KU}$~--- the set of edges connecting every vertex of $U$ with every vertex of $K$;\\
$E_{U} = \{(u_0, aux_{u_0, u_1}), (aux_{u_0, u_1}, u_1), (u_0, aux_{u_0, u_2}), (aux_{u_0, u_2}, u_2),
(u_1, aux_{u_1, u_2}),\\ (aux_{u_1, u_2}, u_2)\}$;

\item $V = \{v_1, \ldots, v_{|V_G|}\}$;\\
$E_{KV}$~--- the set of edges connecting every vertex of $V$ with every vertex of $K$;\\
$E_{VU} = \displaystyle\bigcup_{v \in V}\bigcup_{u \in U}\{(v, aux_{v, u}), (aux_{v, u}, u)\}$;\\
$E_{V} = \displaystyle\bigcup_{v_i, v_j: (i, j) \in E_G}\{(v_i, aux_{v_i, v_j}), (aux_{v_i, v_j}, v_j)\}$.

\end{itemize}

The edge lengths are assigned as follows:

\begin{itemize}

\item $e \in E_K \cup E_{KU} \cup E_{KV} \Rightarrow w_H(e) = 1$;

\item $w_H(u_0, aux_{u_0, u_1}) = w_H(u_0, aux_{u_0, u_2}) = w_H(u_1, aux_{u_1, u_2}) = a_{uu}$,\\
$w_H(aux_{u_0, u_1}, u_1) = w_H(aux_{u_0, u_2}, u_2) = w_H(aux_{u_1, u_2}, u_2) = b_{uu}$;

\item for all $v\in V$, $u \in U$: $w_H(v, aux_{v, u}) = a_{uv}$, $w_H(aux_{v, u}, u) = b_{uv}$;

\item for all pairs $v_i, v_j \in V$ such that $(i, j) \in E_G$: $w_H(v_i, aux_{v_i, v_j}) = a_{vv}$,\\
$w_H(aux_{v_i, v_j}, v_j) = b_{vv}$.

\end{itemize}

\begin{theorem}

Let $G = (V_G, E_G)$ and $H$ = 3-COLORING-$\mathbb{R}^d$-EMBEDDABILI\-TY-REDUCTION$(G)$ described above. Then:

\begin{itemize}

\item If there is no valid 3-coloring of vertices of $G$, then there is no embedding of $H$ in $\mathbb{R}^d$.

\item If a valid 3-coloring of vertices of $G$ exists, then there exists a non-critical embedding of $H$ in $\mathbb{R}^d$.

\end{itemize}

\end{theorem}

\begin{proof}

Consider an embedding of the graph $H$ in $\mathbb{R}^d$; let us construct a valid 3-coloring of vertices of $G$. In every embedding of $H$ all vertices of the set $V \cup U$ are located on some circle of radius $r_0$; denote this circle $\rho$ and its center $O$.

Successively apply the first part of Lemma 4 with the following parameters.

\begin{itemize}

\item Let $v = u_0$, $u = u_1$, $z = aux_{u_0, u_1}$, $l = \text{chord}(2\pi/3 - \varepsilon/2)$, $L = R = \text{chord}(2\pi / 3)$, $r = \text{chord}(2\pi/3 + \varepsilon/2)$.\\
We obtain that $\text{chord}(2\pi / 3 - \varepsilon/2) < ||u_0 -  u_1|| < \text{chord}(2\pi / 3 + \varepsilon/2)$, which is equivalent to $2\pi / 3 - \varepsilon /2 < \angle u_0 O u_1 < 2\pi / 3 + \varepsilon / 2$. We can establish similar inequalities for $u_0$, $u_2$ and $u_1$, $u_2$.

\item Let $v \in V$, $u \in U$, $z = aux_{v, u}$, $l = \text{chord}(\pi / 3 - \varepsilon)$, $L = \text{chord}(\pi/3 - \varepsilon / 2)$, $R = 2r_0$, $r = \infty$.\\
Then $\text{chord}(\pi / 3 - \varepsilon) < ||v - u||$, which is equivalent to $\pi / 3 - \varepsilon < \angle v O u$.

\item Let $v_i, v_j \in V$, $(i, j) \in E_G$: $v = v_i$, $u = v_j$, $z = aux_{v_i, v_j}$, $l = \text{chord}\left(\frac{5}{2}\varepsilon\right)$, $L = \text{chord}(2\pi / 3 - \varepsilon)$, $R = 2r_0$, $r = \infty$.\\
Then $\text{chord}\left(\frac{5}{2}\varepsilon\right) < ||v_i - v_j||$, which is equivalent to $\frac{5}{2}\varepsilon < \angle v_i O v_j$.

\end{itemize}

Construct the coloring of vertices of $G$ as follows: if the vertex $v_i$ lies on the shortest arc between $u_0$ and $u_1$ in the embedding of $H$, the vertex $i \in V_G$ is assigned with the color $c(i) = 2$; if $v_i$ lies on the shortest arc between $u_0$ and $u_2$, then $c(i) = 1$;
otherwise, $c(i) = 0$. It is clear that this coloring is unambigiously defined for any embedding of $H$. We now prove that this coloring of vertices of $G$ is valid, that is, for every edge $(i, j) \in E(G)$ we have $c(i) \neq c(j)$.

Let us show that for every edge $(i, j) \in E_G$ the shortest arc of $\rho$ between the points $v_i$ and $v_j$ contains at least one vertex of $U$. Assume the contrary,
then WLOG both vertices $v_i$ and $v_j$ lie on the shortest arc between $u_0$ and $u_1$, and
$\angle u_0 O u_1 = \angle u_0 O v_i + \angle v_i O v_j + \angle v_j O u_1 > (\pi / 3 - \varepsilon) + \frac{5}{2}\varepsilon +
(\pi / 3 - \varepsilon) = 2\pi / 3 + \varepsilon / 2$. But that contradicts with $\angle u_0 O u_1 < 2\pi/3 + \varepsilon / 2$, thus at least one vertex of $U$ must lie between $v_i$ and $v_j$. In that case the colors of $i$ and $j$ are different; therefore the coloring is valid.
The first part of Theorem 3 is thus proven.

Denote $H' = (V_{H'}, E_{H'}, w_{H'})$, where $V_{H'} = K \cup U \cup V$, $E_{H'} = E_K \cup E_{KU} \cup E_{KV}$, $w_{H'} \equiv 1$. Clearly, $H'$ is a subgraph of $H$.

Now consider a valid vertex 3-coloring of $G$; let us construct a non-critical embedding of $H$.
First, construct a non-critical embedding of $H'$ as follows:

\begin{itemize}

\item choose an arbitrary regular $(d-2)$-simplex with edge length 1 and identify its vertices with vertices of $K$; let $O$ denote the center of the simplex and $\rho$ denote the locus of all points at the distance 1 from all vertices of the simplex; clearly, $\rho$ is a circle of radius $r_0$;

\item choose an arbitrary equilateral triangle inscribed in $\rho$; place the vertices $u_0$, $u_1$, $u_2$ at the vertices of the triangle;
denote $\gamma_0$ the set of all points $x \in \rho$ such that $\angle u_0 O x > \pi - \varepsilon/2$;
clearly, $\gamma_0$ is an open arc of angular measure $\varepsilon$; similarly define sets $\gamma_1$, $\gamma_2$;

\item suppose the vertex $i \in V_G$ is assigned with color $c(i) \in \{0, 1, 2\}$ in the given 3-coloring; place every vertex $v_i \in V$ in such a way that $v_i$ lies on the arc $\gamma_{c(i)}$ for every $i \in V_G$ and no two vertices of $V$ are at the same point; since the arcs $\gamma_0$, $\gamma_1$, $\gamma_2$ have non-zero angular measure, such arrangement of vertices of $V$ is possible.

\end{itemize}

It can be easily verified that the arrangement of vertices of $K \cup V \cup U$ described above yields a non-critical embedding of the graph $H'$.

Now let us add vertices of the set $Aux$ one by one and successively apply the second part of Lemma 4 to show the existence of a non-critical embedding for every new graph. When all vertices of $Aux$ are added, we obtain a non-critical embedding of the graph $H$ since every vertex of $Aux$ is adjacent to exactly two vertices of $V \cup U$.

Successively apply the second part of Lemma 4 with the following parameters.

\begin{itemize}

\item Let $v = u_0$, $u = u_1$, $z = aux_{u_0, u_1}$, $l = \text{chord}(2\pi/3 - \varepsilon/2)$, $L = R = \text{chord}(2\pi / 3)$, $r = \text{chord}(2\pi/3 + \varepsilon/2)$.\\
The points $u_0$ and $u_1$ are at the vertices of an equilateral triangle inscribed in the circle $\rho$, thus $||u_0 - u_1|| = \text{chord}(2\pi / 3)$ and the conditions of the lemma are satisfied.\\
Apply the lemma in a similar way to $u_0$, $u_2$, $aux_{u_0, u_2}$ and $u_1$, $u_2$, $aux_{u_1, u_2}$.

\item Let $v \in V$, $u \in U$, $z = aux_{v, u}$, $l = \text{chord}(\pi / 3 - \varepsilon)$, $L = \text{chord}(\pi/3 - \varepsilon / 2)$, $R = 2r_0$, $r = \infty$.\\
There is at least one vertex $u_i \in U$ such that $\angle v O u_i >\pi - \varepsilon/2$, thus $\angle v O u \geq |\angle v O u_i - \angle u_iOu| > \pi / 3 - \varepsilon / 2$ and $||v - u|| > \text{chord}(\pi / 3 - \varepsilon / 2)$; the conditions of the lemma are satisfied.

\item Let  $v_i, v_j \in V$, $(i, j) \in E_G$: $v = v_i$, $u = v_j$, $z = aux_{v_i, v_j}$, $l = \text{chord}(\frac{5}{2}\varepsilon)$, $L = \text{chord}(2\pi / 3 - \varepsilon)$, $R = 2r_0$, $r = \infty$.\\
The points $v_i$ and $v_j$ lie on different arcs $\gamma_{c(i)}$, $\gamma_{c(j)}$. Let $u_k$ denote the vertex of $U$ that lies on the shortest arc between $v_i$ and $v_j$. Then $\angle v_i O v_j = \angle v_i O u_k + \angle u_k O v_j > 2 (\pi / 3 - \varepsilon / 2) = 2\pi / 3 - \varepsilon$, and $||v_i - v_j|| > \text{chord}(2\pi / 3 - \varepsilon)$; the conditions of the lemma are satisfied.

\end{itemize}

After all applications of Lemma 4 we obtain a non-critical embedding of the graph $H$.

\end{proof}

The constructed graph $H$ has $O(|V_G| + |E_G|)$ vertices and edges (we recall that the dimension $d$ is a fixed constant), but it contains edges of non-unit length; however, for every such edge its length is equal to $\text{RodLength}(a, b)$ for some $a, b$; moreover, the set of possible pairs $(a, b)$ is finite and independent on the input graph $G$. Thus, upon multiple applications of Lemma 2 each edge of non-unit length can be replaced by a subgraph that is isomorphic to $\text{Rod}(a, b)$ for some $(a, b)$; the size of the graph will increase by at most $K$ times, where $K$ is the maximal size of $\text{Rod}(a, b)$ for all used pairs of $(a, b)$; clearly, the value of $K$ depends only on $d$. Therefore the resulting graph $H'$ = 3-COLORING-$\mathbb{R}^d$-UNIT-DISTANCE-EMBEDDABI\-LI\-TY-REDUCTION$(G)$ has $O(|V_G| + |E_G|)$ vertices and edges as well. Finally, we obtain

\begin{theorem}

Let the graph $H'$ = 3-COLORING-$\mathbb{R}^d$-UNIT-DISTANCE-EMBEDDA\-BI\-LI\-TY-\-RE\-DUCTION$(G)$ be constructed by a given graph $G = (V_G, E_G)$ as described above. Then:

\begin{itemize}

\item If there is no valid 3-coloring of vertices of $G$, then there is no embedding of $H'$ in $\mathbb{R}^d$.

\item If a valid 3-coloring of vertices of $G$ exists, then there exists a non-critical embedding of $H'$ in $\mathbb{R}^d$.

\end{itemize}

\end{theorem}

From Theorem 4, the linearity of the size of $H'$, and the fact that the problem of vertex 3-coloring is NP-hard (see  \cite{karp1972reducibility}) Theorem 1 eventually follows.

\pagebreak

\bibliographystyle{ieeetr}
\bibliography{erdos1946sets,de1951colour,karp1972reducibility,chilakamarri1993unit,saxe1980embeddability,horvat2011computational,lovasz1983self,raigorodskii2010chromatic,raigorodskii2012chromatic,niven1956irrational,brass2005research,raigorodskii2013coloring,raigorodskii2014cliques,raigorodskii2001borsuk}

\end{document}